\documentclass[a4paper,10pt]{article}
\usepackage{graphicx}
\usepackage{amssymb,amsmath}
\usepackage{amsthm}
\usepackage{textcomp}       
\usepackage{indentfirst} 
\usepackage{url}
\usepackage{color}
\usepackage{marvosym}
\usepackage{picins} 
\usepackage{fullpage}
\usepackage{authblk}
\usepackage{multirow}
\usepackage{chngpage}

\usepackage[caption=false,font=footnotesize]{subfig}

\newtheorem{lemma}{Lemma}

\newcommand{\revs}[1]{{\color{black}{#1}}}

\begin{document}
 
\title{Rigorous and Practical Proportional-fair Allocation for Multi-rate Wi-Fi}

\author[1]{Paul Patras} 
\author[2]{Andr\'es Garcia-Saavedra}
\author[3]{David Malone}
\author[2]{Douglas J. Leith}
\affil[1]{School of Informatics, The University of Edinburgh, UK}
\affil[2]{School of Computer Science \& Statistics, Trinity College Dublin, Ireland}
\affil[3]{Hamilton Institute, National University of Ireland Maynooth, Ireland}

\date{}

\maketitle

\begin{abstract}
Recent experimental studies confirm the prevalence of the widely known \emph{performance anomaly} problem in current \mbox{Wi-Fi} networks, and report on the severe network utility degradation caused by this phenomenon. Although a large body of work addressed this issue, we attribute the refusal of prior solutions to their poor implementation feasibility with off-the-shelf hardware and their imprecise modelling of the 802.11 protocol. Their applicability is further challenged today by very high throughput enhancements (802.11n/ac) whereby link speeds can vary by two orders of magnitude. Unlike earlier approaches, in this paper we introduce the first rigorous analytical model of 802.11 stations' throughput and airtime in multi-rate settings, without sacrificing accuracy for tractability. We use the proportional-fair allocation criterion to formulate network utility maximisation as a convex optimisation problem for which we give a closed-form solution. We present a fully functional light-weight implementation of our scheme on commodity access points and evaluate this extensively via experiments in a real deployment, over a broad range of network conditions. Results demonstrate that our proposal achieves up to 100\% utility gains, can double video streaming goodput and reduces TCP download times by 8x.
\end{abstract}

\section{Introduction}

\begin{figure}[t]
 \centering
 \includegraphics[width=0.8\columnwidth]{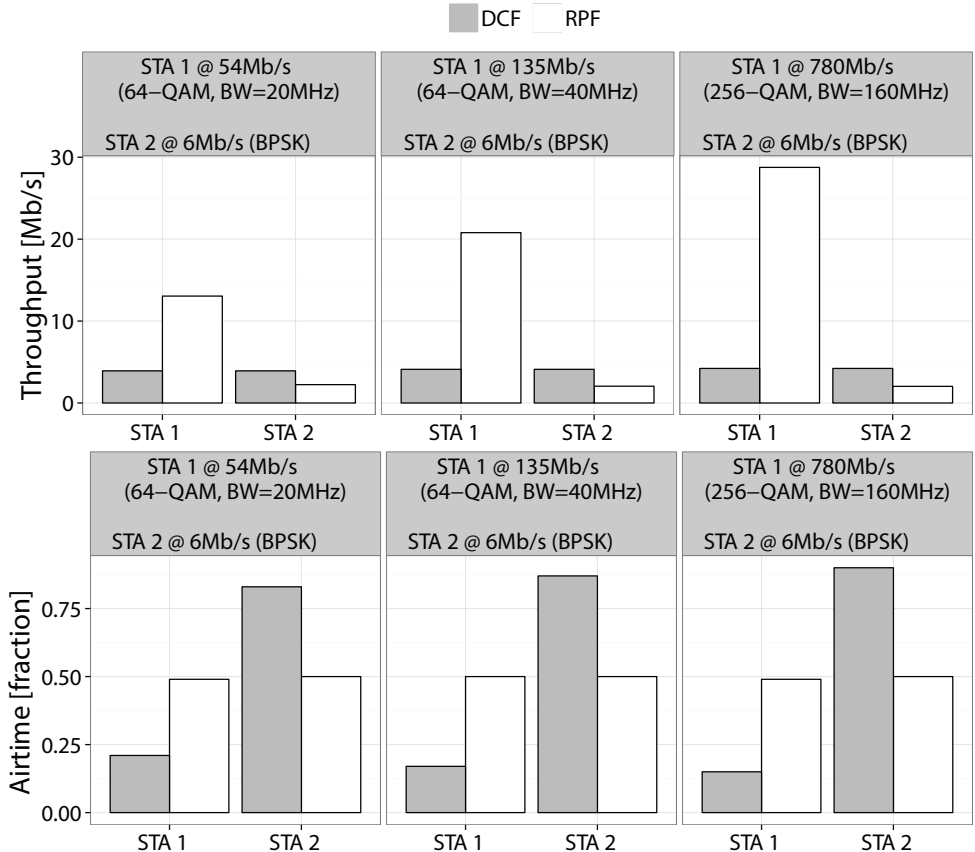}
 \caption{WLAN with 2 backlogged stations sending 1000-byte packets. STA~2 transmits at 6Mb/s (BPSK) on a 20MHz channel. STA 1 transmits at 54Mb/s (20MHz channel, 64-QAM), 135Mb/s (40MHz channel, 64-QAM), and respectively at 780Mb/s (160MHz channel, 256-QAM). Individual throughput (above) and airtime (below) with DCF (dark bars) and respectively with the proposed rigorous proportional-fair (RPF) allocation scheme (light bars). Analytical model.}
 \vspace*{-0.5em}
 \label{fig:example}
\end{figure}

Wi-Fi technology has seen remarkable uptake in recent years through leveraging unlicensed spectrum and employing a simple decentralised channel access paradigm \cite{abi14}. Specifically, client stations that follow the IEEE 802.11 specification use a distributed coordination function (DCF) to schedule their transmission, implementing a medium access control (MAC) protocol that assigns equal transmission opportunities to contenders \cite{IEEE80211rev}. By design, this method ensures stations receive similar throughput, irrespective of their individual link qualities, which is desirable when clients experience good channel conditions. In practical deployments, however, stations employ different modulation and coding schemes (MCSs) to preserve transmission robustness as link qualities change. MCSs of lower indexes handle channel errors better, but have inferior spectral efficiency, which yields lower bit rates. Thus in multi-rate scenarios, stations that run DCF and transmit at lower rates will retain access to the channel for longer periods of time. This degrades the overall network utility, since a significant fraction of channel airtime would be used more efficiently if allowing more frequent transmissions of faster clients. This pathological behaviour is widely known as \emph{performance anomaly} and was first identified with the adoption of high-rate 802.11b services \cite{heusse03}. The effect is dramatically exacerbated when stations that make use of very high throughput protocol enhancements are present in the network. Despite being able to transmit at bit rates up to e.g. 780Mb/s (with 802.11ac \cite{IEEE80211ac}, single stream over 160MHz channel, using 256-QAM), their performance is capped at that of the slowest station in the wireless LAN (WLAN), as we illustrate with dark bars in the example shown in Fig.~\ref{fig:example}. Also plotted in the figure with light bars is the the network performance when the access point (AP) implements the rigorous proportional-fair (RPF) allocation scheme introduced in this paper, demonstrating its effectiveness in providing faster users with significantly higher throughput, while ensuring all stations receive equal airtime. \revs{This contrasts sharply with the legacy DCF, which assigns highly unfair shares of the available channel time to stations that operate with different bit rates (bottom sub-plot).}

The proportional-fair allocation criterion is that which maximises the network utility\revs{\footnote{In this paper we use the same definition for the network utility as given by F. Kelly, i.e. the sum of the logarithms of individual throughputs \cite{kelly97}. \revs{This is appropriate for the \emph{elastic data} applications we consider herein, e.g. file transfer, electronic mail, web browsing, etc. \cite{shenker95}.}}} while providing stations with good airtime fairness \cite{kelly97}. Proportional fairness in IEEE 802.11 WLANs has been extensively studied in the context of high- and extended-rate services, i.e. 802.11b and 802.11a/g \cite{le12,le13,checco11,tychogiorgos12,laddomada10,banchs07,jiang05,li08,siris06}. 
Previous proposals that target equal airtime allocation, however, suffer from at least one of the following key limitations: \emph{(i)} their underlying analyses make several simplifying assumptions that do not capture accurately the 802.11 protocol details, \emph{(ii)} assume error-free channel conditions, and \emph{(iii)}~cannot be implemented with existing off-the-shelf devices, as they require hardware/firmware modifications or changes to the 802.11 state machine and thus are not standard compliant. 

In contrast to earlier works, in this paper we provide, to the best of our knowledge, the first rigorous analysis of 802.11 stations' throughput and airtime, that does not sacrifice accuracy for tractability. \revs{Precisely, our model captures accurately the 802.11 protocol operation and all the possible transmission outcomes in multi-rate settings, allowing for different packet lengths, channel bandwidths, and PHY bit rates, while also accounting for dissimilar data link frame error rates, as typically encountered in practice}. We argue that considering different packet durations is particularly important in practice, as these directly impact on collisions and thus on the fraction of time the channel can be used for successful transmissions. \revs{Capturing all protocol details was often avoided to make 802.11 performance analysis tractable (see e.g. \cite{checco11,banchs07}). We show, however, that the contrary holds and based on our analysis} we formulate proportional-fair allocation as a convex optimisation problem for which we give a closed-form solution. Our approach finds the \emph{optimal contention window} configurations 802.11 stations must employ to maximise utility in multi-rate WLANs.\footnote{NB: Employing the transmission opportunity (i.e. frame bursting) feature to equalise \emph{success airtimes} is not effective for proportional fairness, as this requires total airtime, including collisions, to be the same across all stations.} To demonstrate the implementation feasibility of our theoretical results, we develop a practical yet very accurate proportional-fair allocation scheme with open-source drivers on commodity APs. The key advantage of our design is that it requires no modifications to clients' protocol stack, but only the AP to compute basic statistics of correctly received frames and distribute the optimal MAC parameters through \emph{unicast beacons}. \revs{To the best of our knowledge, our prototype is the first of its kind to achieve utility figures closed to the theoretical optima, while being standard compliant and amenable to incremental deployment on commodity off-the-shelf (COTS) Unix/Linux based wireless routers and access points.}

We evaluate the performance of our solution by conducting extensive experiments in a real 802.11 network over a broad range of multi-rate settings, types of traffic and channel conditions. The results obtained demonstrate that our approach improves network utility by up to 100\%, can double the goodput attained by video streaming applications and can reduce the transfer time of TCP downloads by up to 8 times. Further, we demonstrate that our scheme adapts quickly to rate control decisions and can equalise the throughput of stations that transmit at the same bit rate, but experience dissimilar performance due to capture effect.

In the remainder of this paper we briefly discuss the limitations of prior works (\S\ref{sec:related}) and then present a rigorous analytical model of multi-rate 802.11 operation (\S\ref{sec:analysis}) based on which we formulate and solve explicitly the proportional-fair allocation task as a convex optimisation problem (\S\ref{sec:allocation}). We detail a practical prototype implementation of our solution (\S\ref{sec:implementation}) and report the results of the experimental evaluation undertaken (\S\ref{sec:evaluation}), before concluding the paper (\S\ref{sec:conclusions}).

\section{Related Work}
\label{sec:related}

Recent experimental studies provide substantial evidence of the prevalence and severity of the rate anomaly problem in current Wi-Fi deployments (see e.g. \cite{Patro:2013,Gupta:2012}). The issue persists despite the large body of work conducted in this space since Heusse \emph{et al.} first analysed this behaviour \cite{heusse03}. We attribute this to the poor implementation feasibility of prior approaches addressing proportional fairness with off-the-shelf hardware and their inaccurate modelling of the 802.11 protocol, shortcomings that we particularly tackle in this paper. Here we briefly summarise the most relevant research efforts and highlight the key advantages of our proposal.
\vspace*{0.5em}

\noindent\textbf{Analytical Modelling:} \cite{tychogiorgos12} introduces a general utility maximisation framework for wireless networks, though the analysis does not capture MAC protocol details and is only validated via simulations. The airtime fairness of 802.11 is investigated in \cite{jiang05}, but protocol overhead is ignored and the impact of collisions is neglected. Similarly, the model in \cite{banchs07} neglects collisions of more than two stations and uses a coarse approximation of the average slot time. A more detailed analysis is given in \cite{checco11}; losses due to channel noise, however, are not considered and the authors assume all collisions to be of equal duration. Le \emph{et al.} \cite{le12,le13} give a simplified model of 802.11 multi-rate operation, underestimating the impact of slower clients on collisions, and design a distributed but non-compliant contention window adaptation scheme that is not formally proven to converge. A more accurate multi-rate 802.11 model is presented in \cite{laddomada10}. The authors attempt to solve a \emph{modified} proportional-fair allocation problem, but do not provide a closed-form solution, as this requires explicit computation of the expected slot time, which proves infeasible.
Li \emph{et al.} \cite{li08} formulate the utility maximisation in multi-rate WLANs as a client association task and give algorithms that solve a relaxed version of the problem. This, however, only applies to multi-AP deployments managed by a single entity. Proportional fairness is also studied in multi-hop scenarios, but their complexity requires several questionable simplifications, while resource allocation is subject to non-trivial estimation of individual loads \cite{chakraborty13} or mixed bias strategies \cite{singh08}. 
\vspace*{0.5em}

\begin{table*}[t]
 \begin{adjustwidth}{-1.1in}{-.8in}  
\caption{Comparison between previous works on proportional-fair allocation and our proposal}
\label{tab:feature_comparison}
\centering
\begin{tabular}{|l|c|c|c|c|c|c|c|c|c|c|}
\hline
 & RPF & \cite{tychogiorgos12} & \cite{banchs07} & \cite{jiang05} & \cite{le12,le13,checco11} & \cite{laddomada10} & \cite{siris06} & \cite{heusse05,grunenberger07} & \cite{lee13} & \cite{Gupta:2012,Hegde:2013}\\
\hline
\hline
Models different link error rates & $\checkmark$ & . & . & . & . & $\checkmark$ & . & . & . & .\\
Models precisely simultaneous collisions & $\checkmark$ & . & . & . & . & $\checkmark$ & . & . & . & .\\
Accounts for different payload sizes & $\checkmark$ & . & $\checkmark$ & $\checkmark$ & $\checkmark$ & $\checkmark$ & . & . & . & .\\
Accounts for dissimilar & \multirow{2}{*}{$\checkmark$} & \multirow{2}{*}{.} & \multirow{2}{*}{$\checkmark$} & \multirow{2}{*}{.} & \multirow{2}{*}{.} & \multirow{2}{*}{$\checkmark$} & \multirow{2}{*}{.} & \multirow{2}{*}{.} & \multirow{2}{*}{.} & \multirow{2}{*}{.}\\
\hspace{1em} transmission \& collision durations &  &  &  &  &  &  &  &  &  &\\
Does not require client side changes & $\checkmark$ & . & . & . & . & . & . & . & . & .\\
Suitable for implementation on the AP & \multirow{2}{*}{$\checkmark$} & \multirow{2}{*}{.} & \multirow{2}{*}{$\checkmark$} & \multirow{2}{*}{$\checkmark$} & \multirow{2}{*}{.} & \multirow{2}{*}{.} & \multirow{2}{*}{$\checkmark$} & \multirow{2}{*}{.} & \multirow{2}{*}{DL\textsuperscript{\dag}} & \multirow{2}{*}{DL\textsuperscript{\dag}} \\
\hspace{1em} without FW/HW modifications &  &  &  &  &  &  &  &  &  & \\
Prototyped with COTS & $\checkmark$ & . & . & . & . & . & $\checkmark$ & $\checkmark$ & $\checkmark$ & $\checkmark$\\
Evaluated experimentally & $\checkmark$ & . & . & . & . & . & $\checkmark$ & $\checkmark$ & $\checkmark$ & $\checkmark$\\
\hline
\end{tabular}
\end{adjustwidth}

\begin{flushright}
  \footnotesize 
  \textsuperscript{\dag}AP-side implementation handles \emph{only downlink} traffic.
 \end{flushright}

\end{table*}

\noindent\textbf{Prototyping \& Experimentation:} \cite{siris06} undertakes an empirical study to find the contention window settings that achieve proportional fairness, but lacks analytical support and is limited to static scenarios. Heusse \emph{et al.} propose to control stations' transmission opportunities based on the observed number of idle slots, to tackle airtime fairness \cite{heusse05}; however, the implementation requires precise time synchronisation and is tightly coupled to specific hardware and a proprietary closed-source firmware \cite{grunenberger07}.
Lee \emph{et al.} implement O-DCF \cite{lee13}, whereby a station's packet rate is adjusted according to the MCS employed, to improve network utility. This approach requires introducing an additional queueing layer between application and driver. Similarly, individual queues are introduced and controlled for each destination in WiFox \cite{Gupta:2012}. ADWISER \cite{Hegde:2013} tackles rate anomaly only in the \emph{downlink}, by introducing a dedicated network entity that performs scheduling before the AP.
\vspace*{0.5em}

\noindent\textbf{Our Contributions:} In contrast to the aforementioned works, in this paper we give an accurate model of 802.11 throughput and airtime, accounting for different packet sizes, channel bandwidths, MCSs, and frame error rates. Based on our analytical results we formulate utility maximisation as a convex optimisation problem that we solve explicitly. Further, we provide a practical and accurate implementation of the proposed resource allocation scheme with commodity APs. \revs{Unlike other approaches, our solution does not require modification to the clients' protocol stack, which makes it deployable with client equipment of any vendor, and is demonstrably effective over a broad range of network conditions. We summarise the key advantages of our proposal (RPF) as compared to previous work in Table~\ref{tab:feature_comparison}, which highlights the comprehensive analytical and practical nature of our contribution.}

\section{Analytical Model}
\label{sec:analysis}
In this section we give a rigorous analysis of the throughput and airtime attained by wireless stations in multi-rate \mbox{Wi-Fi} networks that operate in the default infrastructure mode, i.e. where all packets are transmitted through the AP. We consider \revs{a single-hop} 802.11 WLAN with $N$ clients able to select from a fixed set of possible PHY bit rates for transmission, for instance \{6, 9, 12, 18, 24, 36, 48 and 54\} Mb/s available with the 5~GHz OFDM PHY layer (802.11a)~\cite{IEEE80211rev}. In our analysis we consider all stations are saturated, i.e. they always have a packet enqueued for transmission, but later also investigate scenarios with non-saturated real-time traffic. Our analysis allows for arbitrary packet sizes and link error probabilities, either due to signal fading and noise or hidden terminal circumstances that may arise in the presence of different BSSs with partially overlapping coverage. We focus on the single AP case, but argue that, with the advent of software defined technologies, our analysis can be easily extended to dense multi-AP networks underpinned by a management scheme, e.g. as in enterprise or university campus deployments \cite{yiakoumis14}.

\subsection{802.11 Operation}
Wireless stations that follow the IEEE 802.11 specification contend for the medium using two key channel access parameters, namely the minimum and the maximum contention windows, CW$_\text{min}$ and CW$_\text{max}=2^m$CW$_\text{min}$ (where $m$ is the maximum backoff stage). Precisely, to transmit a packet a station will initialise a backoff counter with a random number uniformly distributed in the [0,CW-1] interval, decrement its value every time slot, and transmit when the counter reaches zero. The contention window is initialised with the CW$_\text{min}$ value upon the first attempt, doubled up to CW$_\text{max}$ upon failures (either due to channel errors or collisions with other simultaneous transmissions), and reset when a frame is delivered successfully. Although the standard defines a set of recommended medium access parameters, it allows the AP to change their values and distribute these periodically to associated stations by means of beacon frames \cite{IEEE80211rev}.

\revs{For the analysis we undertake below, it may prove useful to refer to Table~\ref{tab:notation} in the Appendix, where we summarise the notation used throughout.}

\subsection{Throughput Analysis}
In our analysis we assume station $i$ contends for the medium with a probability $\tau_i$ in a randomly chosen slot. To achieve a particular $\tau_i$ in practice, we configure a station with $W_i=$ CW$_{\text{min},i}$ = CW$_{\text{max},i}$, i.e. $m_i = 0$ where, following~\cite{bianchi00},\revs{\footnote{\revs{While this holds for saturation conditions, different packet arrival rates could be considered, e.g. it would be possible to incorporate the probability $q_i$ that a packet is available for transmission when station $i$ wins a transmission opportunity, as suggested in \cite{leith10}, which gives $W_i = (2q_i-\tau_i)/\tau_i$.}}}
\begin{equation}
  W_i = \frac{2 - \tau_i}{\tau_i}.
  \label{eq:tau}
\end{equation}
Practical issues of rounding $W_i$ are addressed in \S\ref{sec:implementation}.

Recall that we account for the fact that nodes may experience different link qualities to the AP and let's denote $p_{n,i}$ the probability that a transmission of station $i$ fails due to channel errors (noise, hidden terminals or interference), which can be estimated practically using e.g. the packet pair technique proposed in \cite{giustiniano07}. Thus, the we can express the conditional failure probability experienced by a station $i$  as follows,
 
\begin{equation*}
  p_{f,i} = 1-(1-p_{n,i})(1-p_i),
\label{eq:failurep}
\end{equation*}
where $p_i$ denotes the collision probability experienced by a packet transmitted by this station and is given by
\begin{equation*}
  p_i = 1-\prod_{j=1,j \neq i}^N (1-\tau_j).
\end{equation*}

The throughput obtained by station $i$ can be expressed as
\begin{equation}
  S_i = \frac{p_{s,i}L_i} {P_eT_e+P_sT_s+P_uT_u},
\label{eq:throughput}
\end{equation}
where $p_{s,i}$ is the probability that the station's transmission is successful and $L_i$ denotes the length of the packet payload generated.\footnote{While here we consider stations send a single packet upon channel access, the current model applies unchanged to 802.11n, where multiple packets are allowed with the same attempt (i.e. frame aggregation).} $P_e$, $P_s$ and $P_u$ are the expected probabilities that a slot is empty (idle), contains a successful, and respectively an unsuccessful transmission (either due to collision or channel errors), while $T_e$, $T_s$ and $T_u$ are their corresponding durations. We provide expressions of the above quantities next.

We compute the probability of a successful transmission of a station $i$ as
\begin{equation}
  p_{s,i} = \tau_i (1-p_{f,i}) = \tau_i (1-p_{n,i})\prod_{j=1,j \neq i}^N (1-\tau_j)
\end{equation}
and the probabilities $P_e$, $P_s$ and $P_u$ as 
\begin{align}
  P_e &= \prod_{i=1}^N (1-\tau_i), \\
  P_s &= \sum_{i=1}^N p_{s,i}, \\
  P_u &= 1-P_e-P_s.
\label{eq:sum_prob}
\end{align}
$T_e$ is a PHY layer constant (e.g. 9$\mu$s for 802.11a/g). In order to calculate the expected durations of a slot containing a success ($T_s$) and respectively a failure ($T_u$), we will index the stations in order of increasing transmission duration. Then
\begin{align*}
  T_s &= \sum_{i=1}^N \frac{p_{s,i}}{P_s}T_{s,i}, \\
  T_u &= \sum_{i=1}^N \frac{p_{u,i}}{P_u}T_{u,i},
\end{align*}
where $p_{u,i}$ is the probability that a slot contains an unsuccessful transmission of stations of highest index $i$, while $T_{s,i}$ and $T_{u,i}$ are the durations of a slot containing a successful and respectively a failed transmission of these stations.

By labelling stations according to their transmission durations and considering the event where a station~$i$ is unsuccessful either due to channel errors or due to a collision with a station of lower index, the probability that a slot contains a failure of a station with highest index $i$ is
\begin{eqnarray}
  p_{u,i} &=& \tau_i p_{n,i} \prod_{j=1,j \neq i}^N (1-\tau_j) + \tau_i \left(1-\prod_{j=1}^{i-1} (1-\tau_j)\right) \prod_{j=i+1}^{N} (1-\tau_j).
\end{eqnarray}
Note that, unlike previous studies, in the above \emph{we do not neglect the probability that more than two stations collide}. Also notice that $P_u = \sum_{i=1}^N p_{u,i}$ and the sum of the probabilities of all the possible slot events expressed in (\ref{eq:sum_prob}) is satisfied.\revs{\footnote{\revs{The duration of a collision is dominated by the frame with the longest duration involved in that collision and collisions should only be counted once. Though by labelling stations according to their transmission duration, we do not eliminate any of the possible collision scenarios, which is verified by the unit sum of the probabilities of \emph{all} the possible slot events.}}}

We compute $T_{s,i}$ as
\begin{equation*}
  T_{s,i} = T_{PLCP} + \frac{H+L_i}{C_i} + SIFS + T_{ack}+DIFS,
\end{equation*}
where $T_{PLCP}$ is the duration of the PLCP (Physical Layer Convergence Protocol) preamble and header, $H$ is the MAC overhead (header and frame check sequence), $C_i$ is the PHY rate employed for transmission (accounting also for wider bandwidth channels), and $T_{ack}$ is the duration of an acknowledgement (ACK). SIFS (Short Inter-frame Space) and DIFS (DCF Inter-frame Space) are PHY layer constants separating a data frame from an ACK and respectively preceding the backoff process (e.g. 16$\mu$s and 34$\mu$s for 802.11a/g). 

Similarly, the duration of a failure involving stations of highest index $i$ is given by
\begin{equation*}
  T_{u,i} = T_{PLCP} + \frac{H+L_i}{C_i} + EIFS,
\end{equation*}
where EIFS is the Extended Inter-frame Space and is a PHY layer constant that is derived from SIFS, DIFS and the time it takes to transmit an ACK frame at the lowest PHY rate, i.e. EIFS = SIFS + DIFS + $T_{ack}(C_{\min})$ \cite{IEEE80211rev}. To reduce notation clutter, we assume that $T_{s,i} \simeq T_{u,i}$, $\forall i$, since the duration of an ACK is dominated by the PLCP preamble and header.
This completes our throughput analysis. 

\subsection{Airtime}

Next, we analyse the \emph{total airtime} of a station, i.e. the fraction of time the channel is occupied by the (successful or unsuccessful) transmission of that station. Denoting $T_{slot}$ the average slot duration ($T_{slot}=P_eT_e+P_sT_s+P_uT_u$), the total airtime used by a station $i$ is given by
\begin{eqnarray*}
  T_i &=& \frac{1}{T_{slot}}\Bigg(\tau_i(1-p_{n,i})\prod_{j=1,j \neq i}^N (1-\tau_j)T_{s,i} + \tau_i p_{n,i} \prod_{j=1,j \neq i}^N (1-\tau_j) T_{s,i} \nonumber \\
 &+& \tau_i \left(1-\prod_{j=1}^{i-1} (1-\tau_j)\right) \prod_{j=i+1}^{N} (1-\tau_j) T_{s,i} + \tau_i\sum_{j=i+1}^{N} \tau_j \prod_{k=j+1}^{N} (1-\tau_k) T_{s,j} \Bigg).
\end{eqnarray*}
\revs{Note that we account for all the possible outcomes of the transmission of a station $i$, i.e. successful reception, failure due to channel errors, collision with at least one station of lower index, and respectively collision with a station of higher index, while we express $T_i$ relatively to the average slot duration $T_{slot}$. After algebra the airtime expression can be reduced to}
\begin{equation}
   T_i = \frac{\tau_i}{T_{slot}}\Bigg( \prod_{j=i+1}^{N} (1-\tau_j) T_{s,i} + \sum_{j=i+1}^{N} \tau_j \prod_{k=j+1}^{N} (1-\tau_k) T_{s,j} \Bigg).
  \label{eq:airtime}
\end{equation}

It is interesting to observe in the above that the airtime expression does not depend on the individual link error probabilities. Let us further rewrite $T_{slot}$ as
\begin{equation*}
 T_{slot}=T_e \prod_{j=1}^N (1-\tau_j)+\sum_{j=1}^N T_{s,j}\tau_j \prod_{k=j+1}^N (1-\tau_k).
\end{equation*}
It is now easy to observe that we give $T_i$ as a function of only stations' transmission attempt probabilities $\tau_i$ and their corresponding successful transmission durations. As we will show later, this has important practical implications on the design of our proportional-fair allocation prototype.

It will also prove useful to work in terms of the transformed variable $x_i=\tau_i/(1-\tau_i)$. Then (\ref{eq:airtime}) becomes
\begin{equation}
 T_i = \frac{x_i}{X}\left( T_{s,i} \prod_{j=1}^{i-1}(1+x_j)+\sum_{j=i+1}^N T_{s,j} x_j \prod_{k=1, k \neq i}^{j-1} (1+x_k)\right),
  \label{eq:airtimeX}
\end{equation}
where
\begin{align*}
X(x)=T_e + \sum_{j=1}^N \left( T_{s,j} x_j \prod_{k=1}^{j-1} (1+x_k)\right).
\end{align*}

The above completes our airtime analysis. In the remainder of this section we establish convexity properties for the throughput expression derived previously, which will allow us to model proportional-fair allocation as a convex optimisation problem and solve this explicitly.

\subsection{Convexity Properties}
\label{sec:convexity}

\revs{To guarantee that the utility maximisation problem we formulate and pursue in this work has a solution that is a global optimum, it is first necessary to establish the convexity properties of the logarithmic utility function $U=\sum_i \log S_i$ associated with the proportional fairness criterion.}

Let us first rewrite the expression of a station's throughput (\ref{eq:throughput}) in terms of the transformed variables  $x_i$. Then 
\begin{align}
  S_i &= (1-p_{n,i})\frac{x_i } {X}L_i
\end{align}
and recall that network utility is the sum of the logarithms of individual throughputs, i.e. $U=\sum_{i=1}^N \log S_i$.

\begin{lemma}\label{lem:one}
Let $f:\mathbb{R}^n \rightarrow \mathbb{R}, f(x)= \sum_{i=1}^n b_ie^{a_i^Tx+c_i}$, where $a_i$ is a column vector, $b_i>0$ and $c_i \ge0$. Then $\log f(x)$ is convex. 
\end{lemma}
\begin{proof}
As $f(x)$ is a sum of exponentials of affine functions, by Boyd and Vandenberghe, p.74 \& p.79 \cite{boyd09} it follows that $\log f(x)$ is convex.
\end{proof}

\begin{lemma}\label{lem:two}
$\log X(e^{\tilde{x}})$ is convex in $\tilde{x}$. 
\end{lemma}
\begin{proof}
Letting $\tilde{x}_i=\log x_i$, we have
\begin{align*}
X(e^{\tilde{x}})=T_e + \sum_{j=1}^N \left( T_{s,j} e^{\tilde{x}_j} \prod_{k=1}^{j-1} (1+e^{\tilde{x}_k}) \right) 
\end{align*}
Expanding the RHS further, it is can be seen that $X$ is of the form $\sum_i b_ie^{a_i^T\tilde{x}+c_i}$. By Lemma \ref{lem:one} it follows that $\log X(e^{\tilde{x}})$ is convex in $\tilde{x}$.   
\end{proof}

Thus the \emph{network utility $U$ can be expressed as an affine combination of convex terms, and therefore is convex}.\revs{\footnote{\revs{Denoting $z_i=1-p_{n,i}$, then $U$ can be expanded to $U = \sum_i (\log z_i + \log x_i + \log L_i - \log X)$. The terms $\log z_i$ and $\log L_i$ are constant, $-\log x_i$ is convex, and by Lemma~\ref{lem:two} $\log X$ is also convex. It follows that the utility is a linear (affine) combination of convex terms, which is convex.
}}}

\section{Proportional-fair Allocation}
\label{sec:allocation}

In this section we formulate proportional-fair allocation as a convex optimisation problem whose objective is network utility maximisation. We solve this problem explicitly and show that the stations' airtimes are equalised at the optimum. We then explain how the solution is used to derive the stations' contention window configurations that achieve this objective.

Let us first denote $z_i=1-p_{n,i}$. The proportional-fair throughput allocation is the solution to the following optimisation problem
\begin{align*}
&\max_{x} \sum_{i=1}^N \log S_i\\
s.t.\ & S_i\le z_i\frac{x_i } {X}L_i,\ i=1,2,\ldots,N; \quad \mbox{(throughput feasibility)}\\
& 0\le x_i,\ i=1,2,\ldots,N. \quad \mbox{($0 \le \tau_i \le 1$)}
\end{align*}
Rewriting the optimisation in terms of the log-transformed variables $\tilde{x}_i=\log x_i$, $\tilde{z}_i=\log z_i$ and $\tilde{S}_i=\log S_i$ the problem becomes
\begin{align}
&\max_{\tilde{x}} \sum_{i=1}^N \tilde{S}_i\\
s.t.\ & \tilde{S}_i-\tilde{z}_i -\tilde{x}_i +\log X -\log{L_i} \le 0,\ i=1,2,\ldots,N. \label{eq:constraint1}
\end{align}
Following the results of \S\ref{sec:convexity}, this optimisation problem is convex, hence a solution exists. The Slater condition is satisfied and so strong duality holds. The Lagrangian is
\begin{align*}
L &=- \sum_{i=1}^N \tilde{S}_i +\sum_{i=1}^N\lambda_i \left(\tilde{S}_i- \tilde{z}_i -\tilde{x}_i +\log X -\log{L_i} \right) \\
\end{align*}
The Karush-Kuhn-Tucker (KKT) condition \cite{hillier09} for $\tilde{S}_{i}$ is
\begin{align*}
 \frac{\partial L}{\partial S_i} = 0,
\end{align*}
which yields
\begin{align*}
 \lambda_i &= 1.
\end{align*}
It follows immediately from complementary slackness that constraint (\ref{eq:constraint1}) is tight at the optimum. Since any optimal solution must lie on the boundary of the log-transformed rate region and, following \cite{leith10}, this is strictly convex, each boundary point has a unique supporting hyperplane, hence, the optimum is unique. 

The solution is obtained by solving the KKT condition for $\tilde{x_i}$, i.e.
\begin{align*}
 \frac{1}{x_i}-\frac{1}{X} \frac{\partial X}{\partial x_i}N = 0, \forall i=1,2,...N.
\end{align*}

Solving the partial derivative, rearranging the above, and using (\ref{eq:airtimeX}), yields
\begin{align}
 T_i=\frac{1}{N}, \forall i.
 \label{eq:airtime_sol}
\end{align}
This means \emph{the solution to the proportional-fair allocation problem assigns equal airtime to all stations, inversely proportional to their total number}. It is also straightforward to observe that the assigned airtimes sum to unity.

With (\ref{eq:airtime}) and (\ref{eq:airtime_sol}) we form a system of $N$ equations and $N$ unknowns (the transmission attempt probability $\tau_i$ of each station), which we can solve for $\tau_i$ numerically. Then, using (\ref{eq:tau}) we can compute the CW configuration of each station to maximise network utility. As we explain next, our AP-based implementation with open-source drivers distributes the computed CWs to clients by means of unicast beacon frames.

\revs{Note that deriving a closed-form solution analytically is a significant result, as the KKT conditions for non-linear programs such as ours cannot be solved directly in most cases. This enables us to develop a practical implementation that solves the optimisation in real-time, taking into consideration the changing network conditions, while running on commodity APs with open-source drivers and requiring minimal computational effort.}

\section{Prototype Implementation}
\label{sec:implementation}

In this section we present a prototype we develop based on our analytical results to enforce proportional-fair allocation in real multi-rate Wi-Fi deployments. \emph{The key advantage of our design is that it only involves light-weight software updates incrementally deployable at the AP and no modifications to the clients' software or hardware}. Further, our approach is modular and thus suitable for future research purposes.\footnote{We intend to make our code available to the community upon publication.} 

We base our implementation on commercial off-the-shelf APs equipped with 802.11 wireless cards and running the Linux {\ttfamily mac80211} subsystem that provides a modular framework for wireless device drivers with fine-grained hardware control \cite{mac80211}. 
Our solution requires only small changes to {\ttfamily mac80211}\footnote{The patch comprising our {\ttfamily mac80211} modifications is \emph{only} 14KB.} and has the major benefit of being independent of any underlying (compliant) hardware, e.g. Atheros, Broadcom, etc. The changes we introduce enable the AP to estimate the average duration of successful transmissions, as well as to distribute the CW configurations to the associated clients, while the optimisation problem is solved in user-space with minimum effort. Note that we run the optimisation periodically (every beacon interval), to ensure our scheme adapts to changes in the WLAN. Our implementation comprises thus two key building blocks: a modified kernel module and an user-space optimisation tool.

\begin{figure}[t]
 \centering
 \includegraphics[width=0.55\columnwidth]{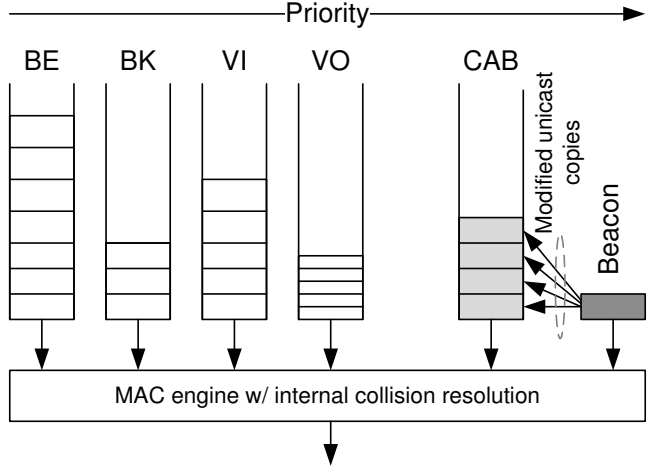}
 \vspace{-0.5em}
 \caption{Software queues exposed by {\ttfamily mac80211} and the implementation principle behind the proposed unicast beaconing mechanism. A copy of the broadcast beacon is made for each client, updated with the corresponding destination address and CW settings, and placed on the CAB queue.}
 \label{fig:queues}
\end{figure}

\subsection{Kernel-space Modifications}

First we develop the low level time sensitive functionality in kernel-space, by augmenting the {\ttfamily mac80211} capabilities to meet our requirements. The AP already collects per-station statistics and we report those in user-space via a dedicated {\ttfamily debugfs}\footnote{RAM-based file system used for debugging purposes, where unlike with {\ttfamily procfs} and {\ttfamily sysfs}, no strict rules apply to information exchange.} space$\leftrightarrow$user-space interface, while all packets pass through the AP in infrastructure operational mode. Therefore, between each iteration the AP aggregates the total time each station spends transmitting, by summing the duration of all the frames received correctly and counting their respective number, then passes this information to user-space using this interface. Similarly, we provide the driver with per-station CW settings that enforce the proportional-fair allocation. To distribute these to the clients, we exploit a standard feature of the 802.11 operation. More precisely, APs periodically broadcast beacon frames to advertise network presence and synchronise associated clients. Beacons can also contain information elements that enforce the contention parameters to be used in the WLAN. We extend this feature and generate \emph{unicast beacons} following the broadcast instance, thus giving each station the CW that maximises network utility. By distributing CWs in this fashion, \emph{we do not require any modifications to the user equipment.}

Note that drivers typically manage six software queues, as shown in Fig.~\ref{fig:queues}. Four of these correspond to different access categories, i.e. best-effort (BE), background (BK), video (VI) and voice (VO), \revs{but as applications and routers tag the DiffServ Code Point (DSCP) field in the packets' IP headers with 0 by default, the BE queue is used for most data traffic.\footnote{\revs{Traffic differentiation at the MAC layer requires appropriate DSCP tagging of the IP packets, which is largely overlooked in practice (see e.g. \cite{juniper14}). Thus data traffic is served by the best effort (BE) queue with the contention parameters we compute.}}} The fifth queue handles multicast frames (referred to as the CAB queue), while the sixth is a dedicated beacon queue. The latter stores a single management frame, whose contents are updated every beacon interval and which is then passed to the hardware for transmission. The frame descriptor is never removed from the queue. Broadcast beacon transmission triggers immediately the CAB queue, which is emptied before handling data traffic. Therefore, we create a copy of the broadcast beacon for each station, change the destination address field to the MAC address corresponding to that client, write the computed CW value into the ECWmin and ECWmax fields of the EDCA (Enhanced Distributed Channel Access) information element (see Fig.~\ref{fig:beacon}), and finally load the unicast beacons into the CAB queue.

\begin{figure}
 \centering
 \includegraphics[width=\columnwidth]{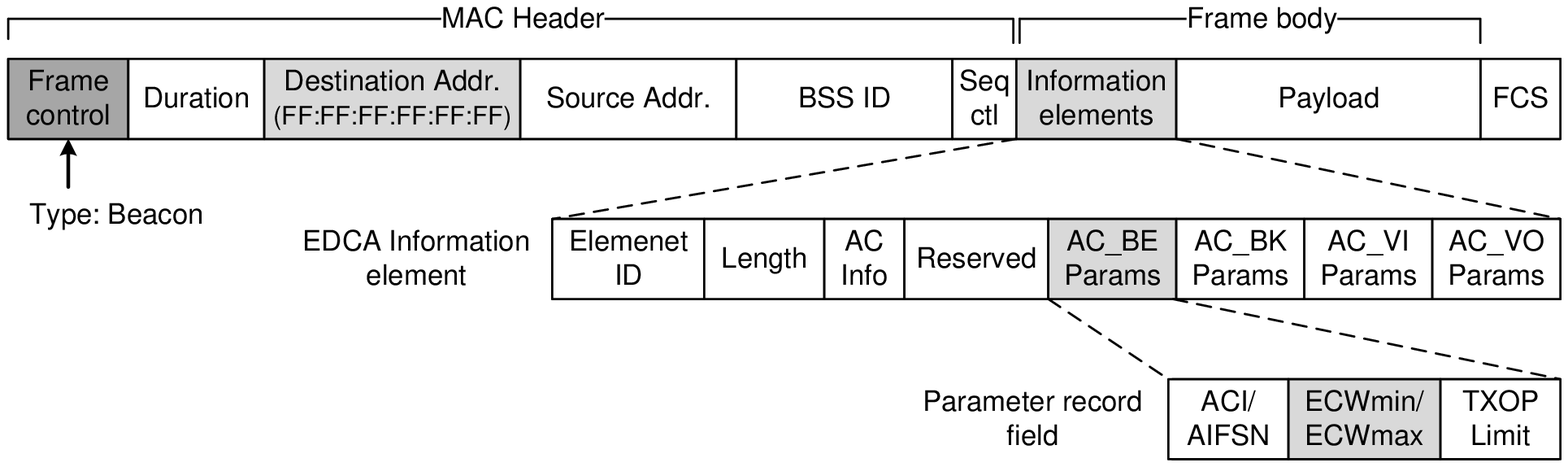}
 \caption{Structure of a beacon management frame. Fields modified to generate unicast beacons that configure individual clients are highlighted. The destination address field is changed from broadcast to that of the target station. ECWmin and ECWmax settings are updated with the values that solve utility maximisation.}
 \label{fig:beacon}
 \vspace*{-0.5em}
\end{figure}

\subsection{User-space Optimisation Tool}

The second component of our implementation is a user-space module that retrieves the measurements from the driver, computes the average duration of a transmission of each station ($T_{s,i}$) and finds the optimal set of CWs that maximise network utility by solving the system of equations given by (\ref{eq:airtime_sol}). The AP handles this task in user-space, since it is more computationally expensive and we do not wish to interrupt driver operations that require responses with microsecond granularity to preserve accurate protocol behaviour. Subsequently, we move the computed CWs to the kernel, to be distributed to each client.

The optimisation tool relies on a simple Python script for exchanging statistics and CW configurations with the driver and uses a GNU Octave\footnote{\ttfamily https://www.gnu.org/software/octave/} 
programme to solve the optimisation task. Specifically, the programme orders stations according to their transmission durations, formalises the system of equations that models the utility maximisation problem and employs a standard iterative solver to compute the CW corresponding to each station. The solutions are rounded to nearest powers of~2 and passed to the driver as $\log_2$CW values.\footnote{While drivers only operate with powers of~2, it would be possible to drive the \emph{average} CW close to the exact value returned by our optimisation, if combined with e.g. sub-gradient methods \cite{valls14}. We leave such refinements for future work.} 

In what follows we demonstrate the effectiveness of our implementation by conducting extensive experiments in a real 802.11 deployment over a broad range of network conditions.

\begin{figure}[t]
\centering
\includegraphics[width=0.7\columnwidth]{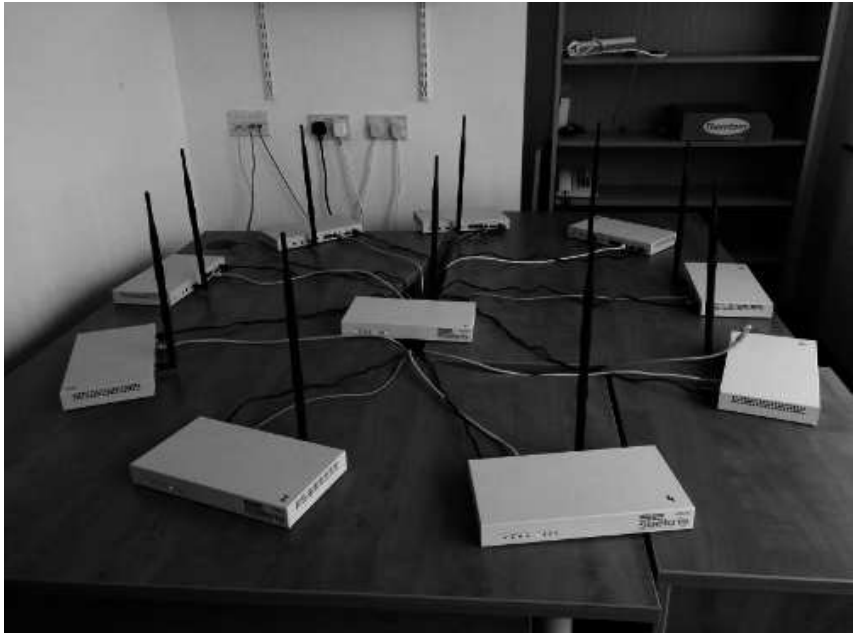}
\caption{Test bed used for the evaluation of the proposed proportional-fair allocation scheme. We use nine Soekris net6501-70 devices, one acting as AP and the others as stations.}
\label{fig:testbed}
\end{figure}

\section{Experimental Evaluation}
\label{sec:evaluation}

In this section we first describe the 802.11 network test bed we deploy to evaluate the proposed proportional-fair allocation scheme, and then report on the performance of the developed prototype under different traffic regimes, rate dynamics and channel conditions. We investigate a broad range of scenarios, ranging from heavily loaded networks, to set-ups whereby users run applications with strict real-time constraints, and circumstances involving lower demands, where the number of active clients and their traffic patterns vary in time. 
We demonstrate that our approach achieves significant gains in terms of throughput, network utility, and latency, over standard DCF operation (CW$_\text{min}$ = 16, CW$_\text{max}$ = 1024) in multi-rate Wi-Fi networks.

\subsection{Test Bed Description}

Our test bed consists of nine Soekris net6501-70 embedded PCs\footnote{\ttfamily http://soekris.com/products/net6501.html} equipped with Compex WLE300NX-6B wireless cards (Atheros AR9390 chipset) that implement the 802.11a/b/g/n specifications, and omni-directional antennas with 8dBi gain. The devices run Ubuntu 14.04 (kernel version 3.13) with {\ttfamily mac80211} and the {\ttfamily ath9k} driver. We use one of the devices as AP and the others as client stations, as shown in Fig.~\ref{fig:testbed}.

We deploy our network in the 5GHz frequency band on channel 149 (centre frequency 5.745GHz). No other deployments are detected in this band at the test bed's location and thus we conclude it is an interference free environment. In our experiments all stations employ the OFDM PHY layer (802.11a) over 20MHz channels, and since all clients are within carrier-sensing range of each other, we disable the RTS/CTS mechanism.

\begin{figure}[t]
\vspace*{-1em}
        \centering
             \subfloat[UDP upload.]{
	    \includegraphics[width=0.48\columnwidth ]{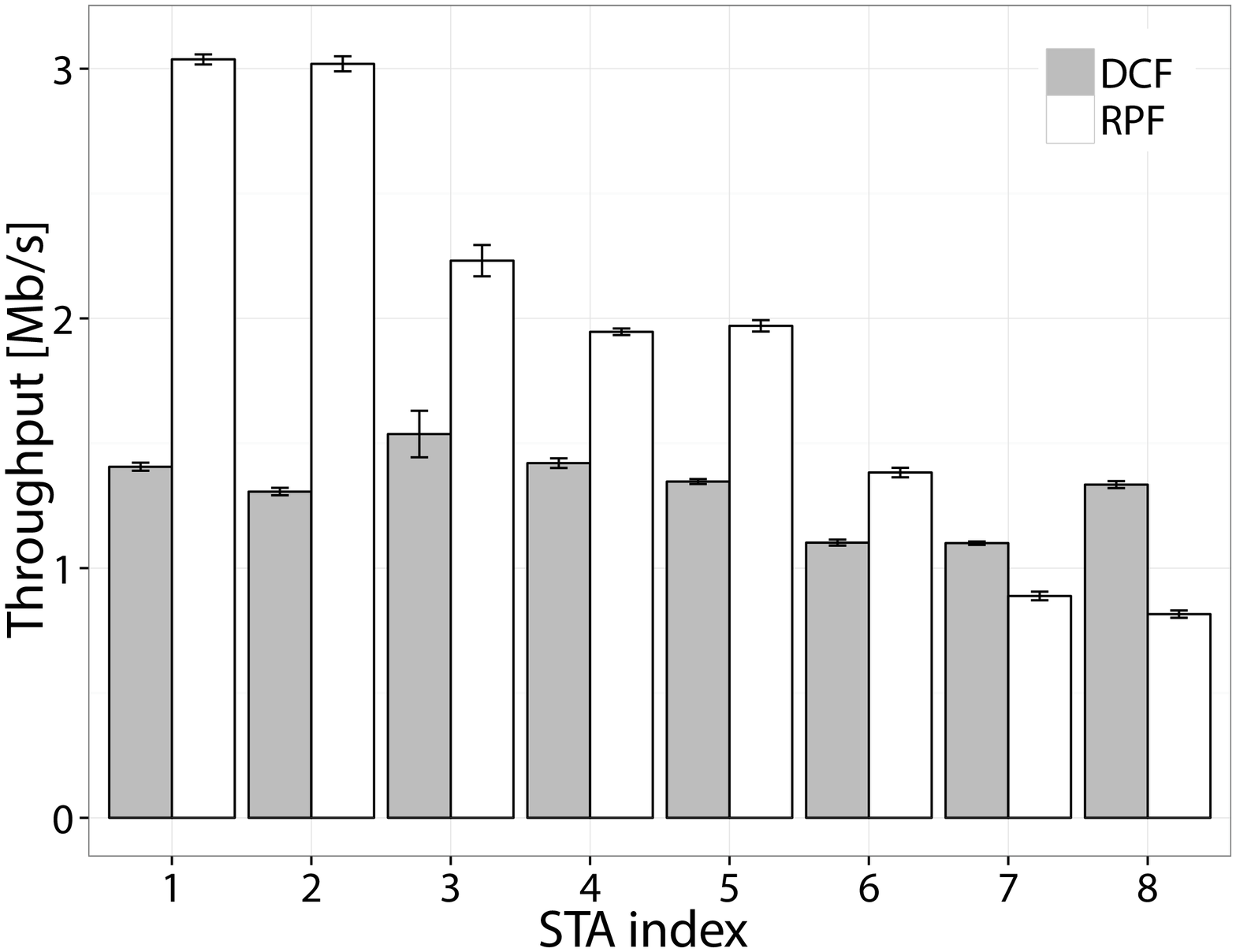}
	    \label{fig:throughput_vs_index:udp}
        }
        \subfloat[TCP upload.]{
                \centering
	    \includegraphics[width=0.48\columnwidth]{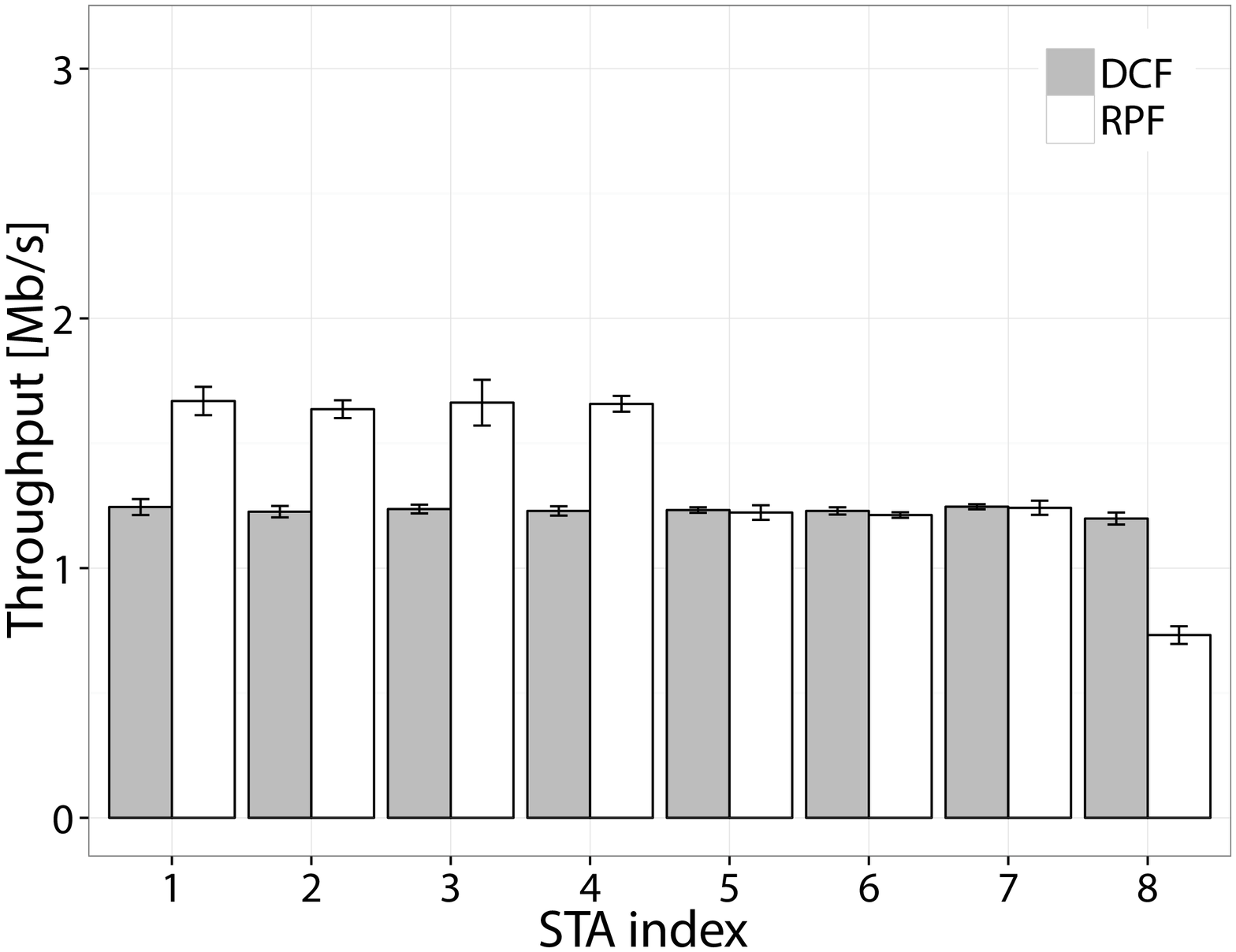}
	    \label{fig:throughput_vs_index:tcp}
        }
  \caption{WLAN with 8 backlogged stations sending 1400-byte packets to the AP. Each station uses the OFDM PHY layer (802.11a) and transmits using a different MCS. Individual throughput performance with standard DCF configuration and the proposed proportional-fair (RPF) allocation scheme over UDP (left) and TCP (right). Experimental data.}
  \label{fig:throughput_vs_index}
\end{figure}

Unless otherwise stated, we run each experiment for a total duration of 60 seconds and repeat the tests 30 times to compute the average and standard deviation of the metrics of interest with good statistical significance.

\subsection{Uplink Data Traffic}

First we consider a scenario with 8 stations, each transmitting at a different bit rate from the set of all possible MCSs, i.e. \{54, 48, 36, 24, 18, 12, 9 and~6\}~Mb/s. We index clients according to the magnitude of the bit rate employed (i.e. STA~1 transmits at 54Mb/s, STA 2 at 48Mb/s, etc.). All stations are backlogged with 1400-byte packets and transmit to the AP first using UDP and later TCP as transport layer protocol. In both cases we examine the individual throughput performance when the network operates with the standard DCF settings and respectively with our proportional-fair (RPF) scheme. The results of these experiments are illustrated in Fig.~\ref{fig:throughput_vs_index}.

\begin{figure}[t]
\vspace*{-1em}
\centering
\includegraphics[width=0.95\columnwidth]{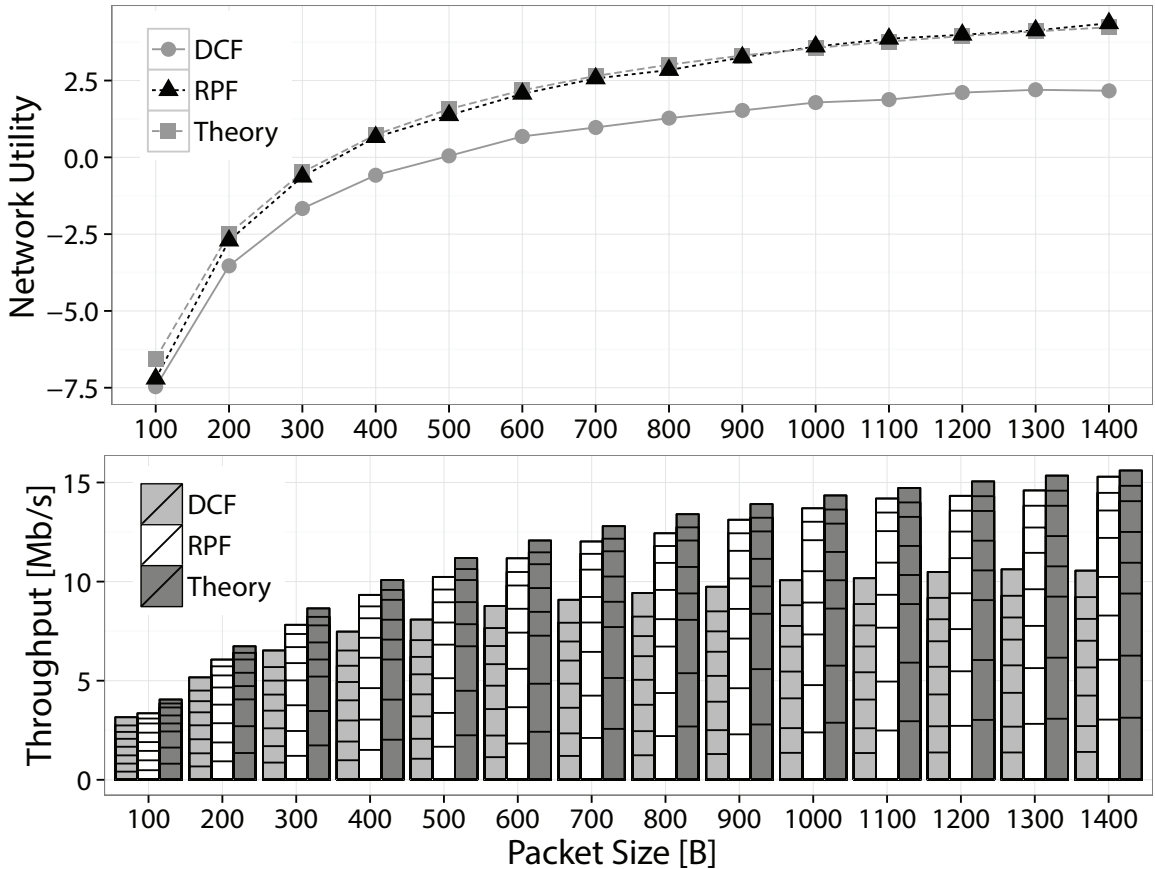}
\vspace*{-0.75em}
\caption{\revs{WLAN with 8 backlogged stations transmitting to the AP using UDP. Stations use the OFDM PHY layer (802.11a) and fully heterogeneous bit rates. Network utility (top) and individual/total throughput (below) with standard DCF configuration and the proposed proportional-fair (RPF) allocation scheme as the packet size is varied. Experimental data vs. theoretical optimum prediction.}}
\vspace*{-0.5em}
\label{fig:throughput_vs_L}
\end{figure}

We observe that when all stations use the same CW configuration, the UDP throughput obtained by each is approximately the same, irrespective of the PHY bit rate used for transmission.\footnote{We observe small differences in the individual performance, which we attribute to the \emph{capture effect}. Note that we investigate this phenomenon more closely in \S\ref{sec:capture}.} Conversely, when the AP assigns the CW configurations obtained with our approach (RPF), faster stations improve their throughput by up to 120\%, while the performance of clients transmitting at inferior rates is only marginally affected (Fig.~\ref{fig:throughput_vs_index}a). In this experiment, our proposal improves the network utility by 100\%.

On the other hand it is interesting to observe that when the packets are sent over TCP, the transport layer acknowledgements in the downlink direction have collision-mitigation and pseudo-scheduling effects. As a consequence, stations experience identical performance with DCF and similar to that of UDP. The overall network utility, however, remains suboptimal, which our approach is able to correct even in this case by deriving the appropriate CW settings (Fig.~\ref{fig:throughput_vs_index:tcp}). Specifically, our scheme provides faster stations with 40\% more throughput, while improving network utility by 40\%.

Next we investigate the impact of frame sizes on both the individual throughput and network utility, again considering a scenario with eight stations transmitting at heterogeneous bit rates and using the UDP transport protocol. To this end, we vary the size of packets sent by each station from 100 to 1400 bytes with 100-byte increments, while examining the network performance with both the default DCF configuration and the proposed proportional-fairness enforcing solution. The results of these experiments are shown in Fig.~\ref{fig:throughput_vs_L}, \revs{where for comparison we also include the theoretical optimum values predicted by our model.}

Note again that DCF gives all stations equal transmission opportunities and thus individual throughputs are capped at that of the slowest contender. Conversely, our approach favours faster clients without starving slower ones and as a consequence improves significantly the network utility, especially as the packet size increases, in which case the utility is doubled. \revs{Furthermore, the performance of RPF in practice is very close to that predicted with our analytical model, even though this is not an idealised network.}

\begin{figure}
\vspace*{-0.15em}
\centering
\includegraphics[width=0.8\columnwidth]{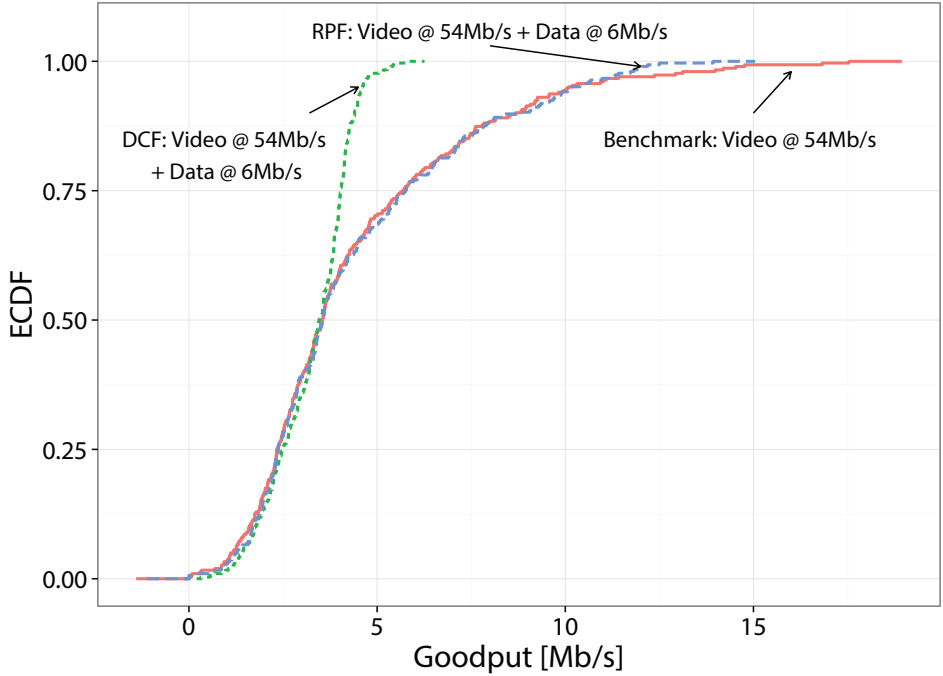}
\vspace*{-0.45em}
\caption{WLAN with 2 stations: one streams MP4 video from the AP over a wireless link operating at 54Mb/s; the other sends backlogged uplink TCP traffic at 6Mb/s PHY bit rate. Empirical CDF of the goodput attained by the video stream alone and with competing traffic, when the AP operates with legacy DCF and respectively with the proposed RPF scheme. Experimental data.}
\vspace*{-0.8em}
\label{fig:video}
\end{figure}

\subsection{Video Streaming}
In what follows we study the performance of our prototype in the presence of real-time traffic. To this end we consider a scenario where one station experiences modest link quality when performing a TCP upload and thus transmits at a 6Mb/s bit rate (background traffic). Simultaneously, a second station is streaming video from the AP over HTTP (using the VLC media player) while operating at a 54Mb/s PHY bit rate. For this experiments we use a $10$-minute fragment of the `Big Buck Bunny' cartoon film, encoded in MP4 at 1280x720.\footnote{Available under (CC) license at {\ttfamily \url{http://www.bigbuckbunny.org}}}

Initially we stream the video without background traffic and use the measured goodput as a benchmark for the subsequent tests. Then we evaluate the performance of the same video stream with the sluggish data station sending traffic, first with the AP advertising the default DCF configuration, and second with the AP running our RPF prototype. For all experiments, we collect samples of the stream's rate every one second and plot the empirical CDF of the attained goodput in Fig~\ref{fig:video}.  

From this figure we observe that, in the presence of sluggish background traffic, the legacy DCF protocol struggles to satisfy the demand of the video stream. In contrast, when enforcing proportionally-fair allocation via our scheme, the capacity of the video link is sufficiently enhanced to ensure perfect quality of the video streaming experience.

\begin{figure}[t]
\centering
\includegraphics[width=1\columnwidth]{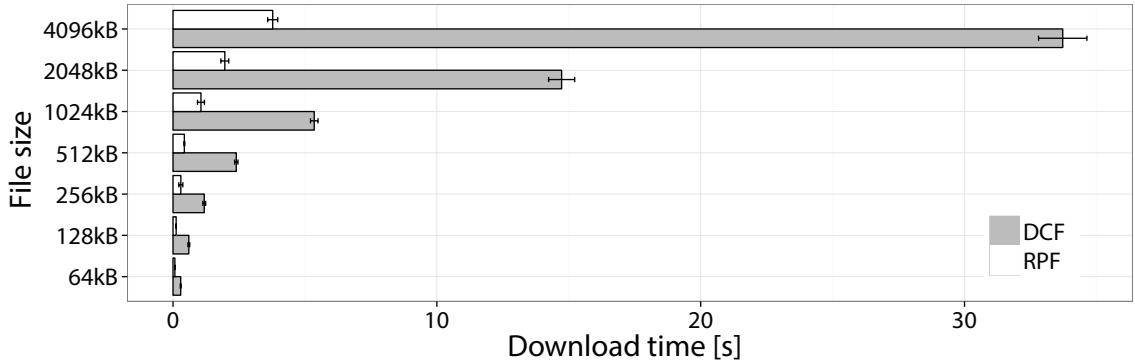}
\vspace*{-0.5em}
\caption{WLAN with 2 stations: one downloads small-sized files from the AP using HTTP over a link operating at 54Mb/s; the other is backlogged and uploads to the AP using TCP over a 6Mb/s link. Download time of different-size objects with DCF and the proposed RPF scheme. Experimental data.}
\label{fig:download}
\end{figure}

\begin{figure}
\centering
\vspace*{-0.75em}
\includegraphics[width=0.9\columnwidth]{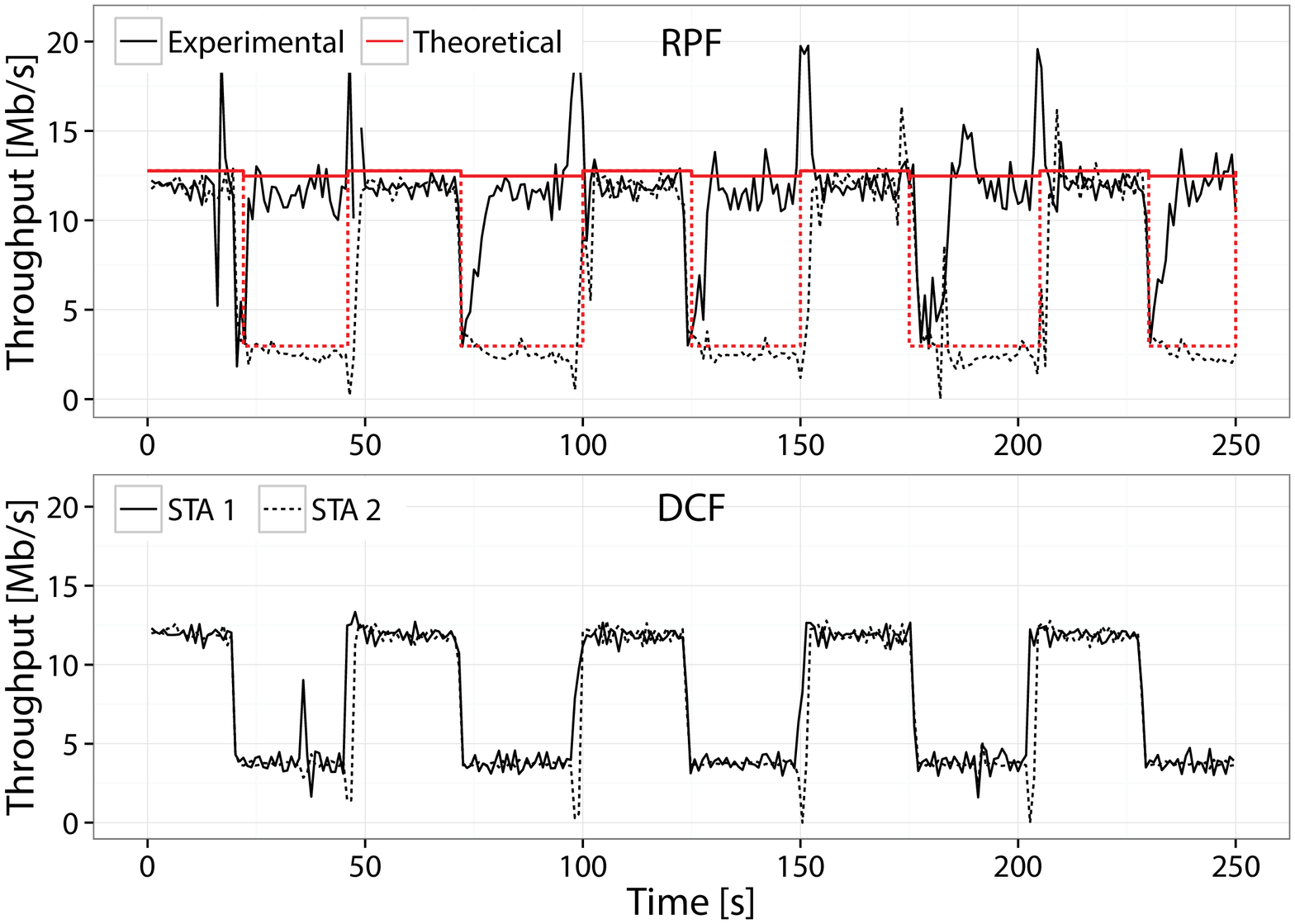}
\vspace*{-0.5em}
\caption{\revs{WLAN with 2 backlogged stations sending UDP traffic to the AP. STA 1 always operates at 54Mb/s. STA 2 alternates its bit rate between 54 and 6 Mb/s every 25 seconds. Time evolution of the individual throughputs with both DCF and the proposed RPF scheme. Experimental data vs theoretical optimum prediction.}}
\vspace*{-0.75em}
\label{fig:mcs_adaption}
\end{figure}

\subsection{Small File Downloads}
Doubtlessly, proportional-fair allocation enhances the utility of wireless networks that exhibit multi-rate links, as we have shown so far for applications that involve transmissions of relatively large volumes of data. Many other popular applications such as \emph{web browsing} or \emph{e-mail}, however, exchange frequent and small information objects. For this sort of applications, the user's primary interest is in downloading these objects with minimum latency. 

To evaluate the performance of our prototype with such Internet applications, in this subsection we consider a scenario whereby one station performs an upload of a large file using TCP and transmitting at 6Mb/s bit rate, while a second client downloads multiple small files over a link that operates at 54Mb/s. More specifically, we configure the AP as an HTTP server where we store a set of files with sizes that range between 64KB and 4096KB, and use the {\ttfamily wget} tool on the client station to retrieve these files.

We repeat this experiment 50 times for each file size and measure the average and standard deviation of the download duration when the network operates with both the default DCF configuration and respectively with our RPF implementation. The corresponding results depicted in Fig.~\ref{fig:download} demonstrate a major reduction of the download times when using the proposed RPF scheme, with latency decreasing by 315\% for the smallest object and by 790\% for the largest one. Note that for such small transfers the underlying TCP congestion control algorithms used by the HTTP transfers remain in the slow-start phase for most of the time and, given that our RPF prototype grants faster stations a larger number of Layer 2 transmission opportunities, the round-trip time (RTT) reduction achieved translates to the high performance gains observed at the application layer.

\subsection{Rate Adaptation}

Current operating systems implement rate control algorithms to dynamically adapt the MCS of wireless stations to variable channel quality (e.g. Linux's {\ttfamily  mac80211} subsystem supports both {\ttfamily minstrel} and a PID, i.e. proportional-integral-derivative, rate control algorithm). Thus in what follows we assess our RPF prototype's ability to track drastic changes of the PHY bit rate used by clients for their transmissions.

To this end we devise a network set up with two stations transmitting backlogged UDP traffic to the AP, initially both at 54Mb/s. While we keep the bit rate of STA 1 fixed for the entire duration of the experiment, STA 2 switches the modulation scheme between 6 and 54Mb/s every 25 seconds, and we examine the time evolution of individual throughputs for a total duration of 250s when the WLAN operates with the default DCF settings and respectively with the proposed RPF scheme. We repeat this experiment several times and present one representative snapshot in Fig.~\ref{fig:mcs_adaption}, while we confirm that the observed behaviour is consistent across all experiments.

Observe in Fig.~\ref{fig:mcs_adaption} that our scheme rapidly detects the changes in the network and adapts the CWs accordingly. The case when STA 2 reduces its bit rate from 54 to 6Mb/s is particularly important as the network suddenly operates sub-optimally with DCF. In contrast, our proposal drives the network close to optimal behaviour whereby utility is maximised in less than 10 seconds. On the other hand when STA 2 switches from 6 to 54Mb/s, we observe a brief overshoot in STA 1's throughput due to a temporary over-allocation, but our scheme corrects this immediately and ensures both contenders experience identical performance thereafter. \revs{Note that we observe similarly quick responses to changes in the network conditions for all our experiments (which include up to 8 stations). Also shown in Fig.~\ref{fig:mcs_adaption} is the theoretical optimum throughput values that stations performing (perfect) rate control are expected to achieve. This confirms that with RPF running on the AP, stations required to reduce their bit rates (due to degrading links) will receive only marginally less throughput than with DCF and close to the theoretical optimum, while the performance of the clients not experiencing changes in the link quality remains largely unaffected and nearly optimal.}

\subsection{Capture Effect}
\label{sec:capture}

In real WLAN deployments the so-called \emph{capture effect} is frequently encountered, allowing a wireless receiver to decode a frame with relatively higher signal strength even in the presence of a collision, provided the difference in the power levels of the concurrent transmissions is sufficiently large~\cite{Patras:2012:WoWMoM}. Indeed, this happens in practice even for (apparently) homogeneous set-ups in terms of AP--client distances and TX powers. This is the case in the experiment reported earlier in Fig.~\ref{fig:throughput_vs_index:udp}. 

\begin{figure}[t]
\centering
\includegraphics[width=\columnwidth]{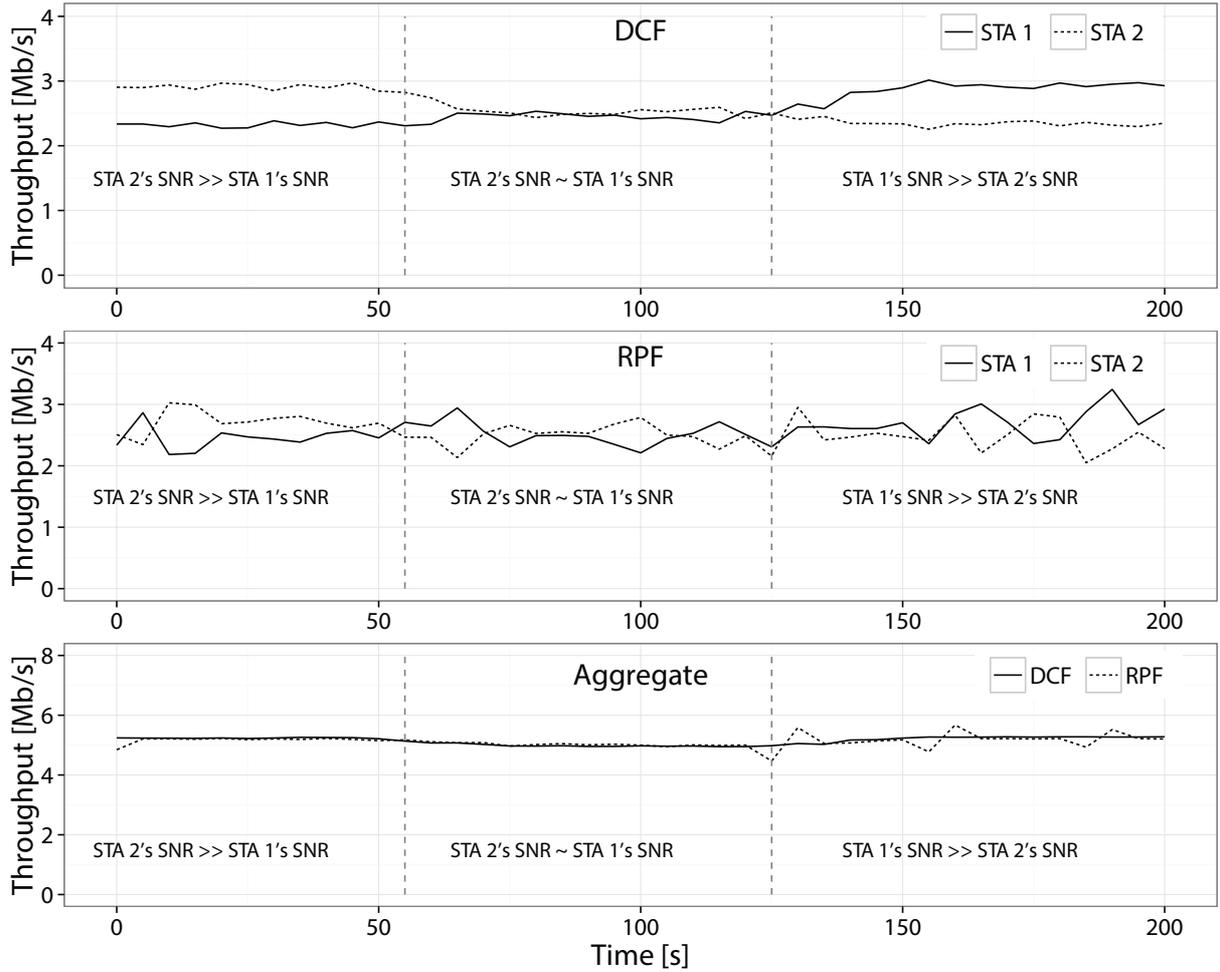}
\vspace*{-0.5em}
\caption{WLAN with 2 stations sending at 6Mb/s. STA 1 uses fixed TX~power (16dBm). STA 2 starts with maximum supported power (30dBm) and reduces this gradually down to the minimum (0dBm). Time evolution of individual throughputs with DCF (top), the proposed RPF scheme (middle), and the aggregate network throughput in both scenarios (bottom). Experimental data.}
\vspace*{-1em}
\label{fig:capture}
\end{figure}

While the capture phenomenon may increase the total throughput of the network \cite{Patras:2012:WoWMoM}, DCF yields an unfair distribution of the individual throughputs, as it provides all competing stations with the same CW settings, without accounting for the fact that some will not double their CW upon collisions and thus attempt more frequently. Specifically, stations delivering frames at larger signal to noise ratios (SNR) will have higher transmission attempt rates and consequently attain superior performance. In contrast to DCF, the proposed RPF allocation scheme targets equalising \emph{airtimes} and, to achieve this, the AP considers all transmissions, including those that captured during collisions, thereby alleviating this effect.

To confirm this, we conduct another experiment whereby two competing clients transmit UDP flows in the uplink direction, both at 6Mb/s. In our experiment, the transmission power of STA 1 remains fixed at 16dBm and we vary STA~2's transmission power while we measure their individual throughputs with both DCF and the proposed RPF scheme. Precisely, STA 2 starts with the maximum supported level (30dBm), thus benefiting initially from capture with DCF, and we gradually reduce this until the SNR of both stations is similar, i.e. no capture effect is observed. Continuing this process the SNR of STA 2 falls below that of STA 1, who now benefits, as illustrated in Fig.~\ref{fig:capture} (top sub-plot).

In contrast to DCF, our RPF scheme provides both stations with the same throughput even in the presence of capture effect (Fig.~\ref{fig:capture}, middle sub-plot), while this does not impact negatively on the total throughput performance of the network, as depicted in the bottom sub-plot of Fig.~\ref{fig:capture}. 

\begin{figure}[t]
\vspace*{-1em}
\centering
\includegraphics[width=\columnwidth]{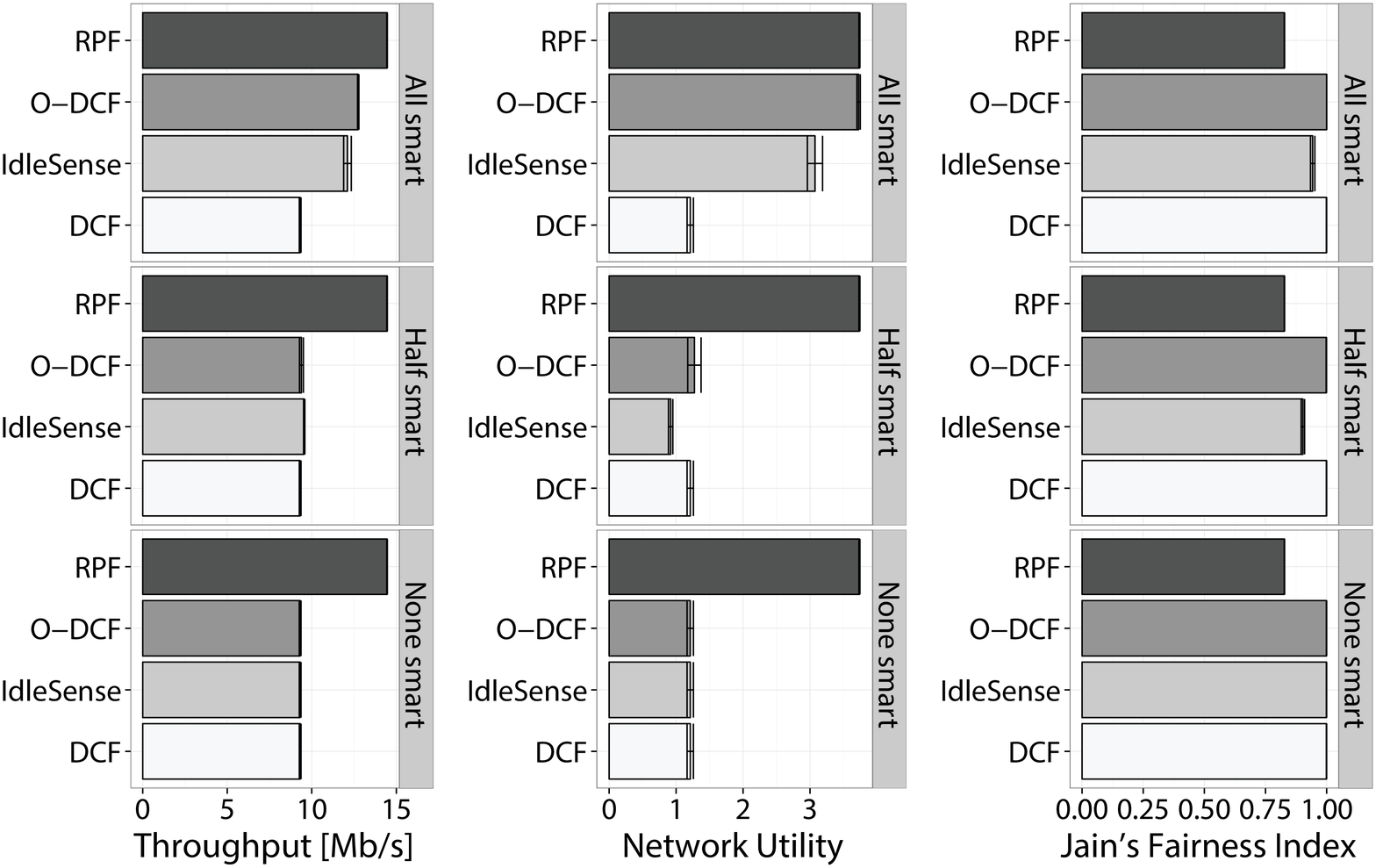}
\vspace*{-0.75em}
\caption{\revs{WLAN with 8 backlogged stations transmitting using UDP. Stations use the OFDM PHY layer (802.11a), fully heterogeneous bit rates and employ our proposal (RPF), Idle Sense, \mbox{O-DCF}, and standard DCF. Total throughput (left), network utility (centre), and Jain's fairness index (right) in three scenarios: all clients smart (implementing changes as required, top), half smart (middle), and respectively all legacy (none smart, bottom). Simulations results.}}
\vspace*{-0.5em}
\label{fig:simulations}
\end{figure}

\revs{ 
\subsection{Comparison against Other Approaches}
As discussed in \S\ref{sec:related}, some of the previously proposed solutions to the performance anomaly problem have been implemented in practice, though they involve non-trivial modifications of the clients' firmware/drivers. Such modifications render a performance comparison through experiments infeasible. Therefore to provide a quantitative assessment of the advantages of the proposed RPF scheme over earlier prototypes, in this section we compare the performance of two such approaches, Idle Sense \cite{heusse05} and O-DCF \cite{lee13}, against the standard DCF and RPF, by means of system level simulations. We note that in contrast to other schemes that target proportional fairness, we have already shown that RPF's performance is very close to the theoretical optimum, therefore an exhaustive comparison would not be justified.

We consider again the case of a wireless LAN with eight backlogged client stations, each transmitting 1000-byte UDP packets at one of the bit rates available with the OFDM PHY layer (802.11a), and measure the total throughput, network utility, and Jain's fairness index \cite{jain84} achieved by each approach. 
In each case, we focus on three different scenarios, namely all stations supporting the proposed mechanism (``All smart''), half of them boasting the necessary modifications (``Half smart''), and respectively all being legacy 802.11 clients (``None smart'').
Remember that for both standard DCF and RPF no client side modifications are required.

We run each simulation for a total duration of 5 minutes, and collect statistics after a 10-second warm up interval. We repeat each simulation 10 times and report average values and 95\% confidence intervals for the metrics of interest. The results are depicted in Fig.~\ref{fig:simulations}, where we observe that our approach maximises network performance in all scenarios and outperforms the other approaches, especially as the number of legacy clients increases. Furthermore, RPF achieves a good compromise between maximum network throughput and fair throughput distribution, as Jain's fairness index is reasonably high.

From the above experiments we conclude that the proposed proportional-fair allocation scheme significantly increases network utility, improves the performance of video streams, reduces small file download times, reacts fast to rate control decisions, and alleviates the capture effect in real 802.11 networks.
}

\section{Conclusions}
\label{sec:conclusions}

In this paper we made the case for accurate modelling of multi-rate 802.11 operation, establishing a rigorous analysis that accounts for different packet sizes, channel bandwidths, PHY bit rates and frame error rates. We formulated network utility maximisation as a convex optimisation problem that we solved explicitly. We designed and implemented a practical scheme that achieves the desired proportional-fair resource allocation only with light-weight modifications to commodity APs' software and no changes to user equipment. Experimental results in a real 802.11 deployment demonstrated substantial performance gains in terms of throughput, utility, video performance, and file download times, over a broad range of channel conditions, client activity levels, and traffic regimes.

\section*{Acknowledgements}
The research leading to these results has received funding from Science Foundation Ireland grant no. 11/PI/1177.

\bibliographystyle{elsarticle-num}

\newpage
\revs{\section*{Appendix}}
 \begin{table*}[h]
\centering
 \revs{
\caption{Summary of notation used throughout the undertaken performance analysis.}
\label{tab:notation}
\centering
\begin{tabular}{|l|l|}
\hline
 Symbol & Interpretation \\
\hline
\hline
$N$ & Total number of stations \\
$L_i$ & Length of packet transmitted by station $i$\\
$C_i$ & Transmission rate of station $i$ \\
$W_i$ & Contention window employed by client $i$ \\
$\tau_i$ & Probability a station $i$ transmits in a randomly chosen slot\\
$p_{f,i}$ & Conditional failure probability seen by station $i$\\
$p_{n,i}$ & Link error probability experienced by station $i$\\
$p_i$ & Collision probability experience by station $i$\\
$p_{s,i}$ & Probability a slot contains a successful TX of station $i$\\
$p_{u,i}$ & Probability a slot contains a successful TX of station $i$\\
$P_e, P_s, P_u$ & Probabilities that a slot is empty, \\
& contains a successful, and unsuccessful transmission\\
$T_{s,i}$ & Duration of a successful transmission of client $i$\\
$T_s, T_u$ & Average duration of a successful and failed transmission\\
$T_e$ & Duration of an empty slot (PHY layer constant)\\
$T_{slot}$ & Average slot duration\\
$T_{PLCP}$ & Duration of the Physical Layer \\
& Convergence Protocol preamble and header\\
$H$ & MAC overhead (header and frame check sequence)\\
$T_{ack}$ &Duration of an acknowledgement frame\\
SIFS, DIFS, & Short, DCF, and Extended Inter-frame Space \\
~~~~EIFS & (PHY layer constants)\\
$S_i$ & Throughput attained by station $i$ \\
$T_i$ & Fraction of the total airtime occupied by station $i$ \\
\hline
\end{tabular}
} 
\end{table*}

\end{document}